\documentclass[DIV=11,a4paper]{artclcomp}

\usepackage{times}
\usepackage{calc}
\usepackage{enumitem}
\usepackage{eurosym}
\usepackage{textcomp}
\usepackage{mathrsfs}
\usepackage{graphicx}
\usepackage{enumitem}
\usepackage{stmaryrd}
\usepackage{epsfig}
\usepackage[usenames,dvipsnames]{xcolor}

\newcommand{\set}[1]{\ensuremath{ \{ #1 \} }}
\newcommand{\R}{\mathbb{R}}

\usepackage[hyperref,amsmath,thmmarks]{ntheorem}

\theoremseparator{.}

\theoremsymbol{\ensuremath{\blacklozenge}}
\theoremheaderfont{\itshape}

\numberwithin{equation}{section}

\usepackage{lipsum}
\usepackage{todonotes}
\usepackage{graphicx}
\usepackage{amsfonts}
\usepackage{accents}
\usepackage{xspace}
\usepackage{dsfont}

\allowdisplaybreaks[1] 

\def\set{\ensuremath{\mathcal{S}}}
\def\bx{\ensuremath{\mathbf{x}}}
\def\bu{\ensuremath{\mathbf{u}}}

\def\FH{Fr\'echet--Hoeffding\xspace}

\def\ud{\mathrm{d}}
\def\cD{\ensuremath{\mathcal{D}}}

\def\P{\ensuremath{\mathbb{P}}\xspace}
\def\X{\ensuremath{\mathbf{X}}\xspace}

\newcommand{\upQ}{\overline{Q}} 
\newcommand{\downQ}{\underline{Q}}

\newcommand{\refC}{C^*} 
\newcommand{\lo}{\preceq}

\DeclareMathAccent{\what}{\mathord}{largesymbols}{"62}

\DeclareFontFamily{U}{mathx}{\hyphenchar\font45}
\DeclareFontShape{U}{mathx}{m}{n}{
      <5> <6> <7> <8> <9> <10>
      <10.95> <12> <14.4> <17.28> <20.74> <24.88>
      mathx10
      }{}
\DeclareSymbolFont{mathx}{U}{mathx}{m}{n}
\DeclareFontSubstitution{U}{mathx}{m}{n}
\DeclareMathAccent{\widecheck}{0}{mathx}{"71}

\begin{document}

\title{Model-free bounds on Value-at-Risk using extreme value information and statistical distances}

\author[a,1,s]{Thibaut Lux}
\author[b,2,s]{Antonis Papapantoleon}

\address[a]{{Department of Finance, Vrije Universiteit Brussel, Pleinlaan 2, 1050 Brussels, Belgium}}
\address[b]{{Department of Mathematics, National Technical University of Athens, Zografou Campus, 15780 Athens, Greece }}

\eMail[1]{tlux@consult-lux.de}
\eMail[2]{papapan@math.ntua.gr}

\myThanks[s]{We thank Peter Bank, Carole Bernard, Fabrizio Durante, Ludger R\"uschendorf, Kirstin Strokorb, Steven Vanduffel and Ruodu Wang for useful discussions during the work on these topics. TL gratefully acknowledges the financial support from the DFG Research Training Group 1845 ``Stochastic Analysis with Applications in Biology, Finance and Physics''. }

\abstract{
We derive bounds on the distribution function, therefore also on the Value-at-Risk, of $\varphi(\mathbf X)$ where $\varphi$ is an aggregation function and $\mathbf X = (X_1,\dots,X_d)$ is a random vector with known marginal distributions and partially known dependence structure.
More specifically, we analyze three types of available information on the dependence structure: 
First, we consider the case where extreme value information, such as the distributions of partial minima and maxima of $\mathbf X$, is available. 
In order to include this information in the computation of Value-at-Risk bounds, we utilize a reduction principle that relates this problem to an optimization problem over a standard Fr\'echet class, which can then be solved by means of the rearrangement algorithm or using analytical results. 
Second, we assume that the copula of $\mathbf X$ is known on a subset of its domain, and finally we consider the case where the copula of $\mathbf X$ lies in the vicinity of a reference copula as measured by a statistical distance. 
In order to derive Value-at-Risk bounds in the latter situations, we first improve the \FH bounds on copulas so as to include this additional information on the dependence structure. 
Then, we translate the improved \FH bounds to bounds on the Value-at-Risk using the so-called improved standard bounds.
In numerical examples we illustrate that the additional information typically leads to a significant improvement of the bounds compared to the marginals-only case.}

\keyWords{Value-at-Risk bounds, dependence uncertainty, copulas, improved \FH bounds, distribution of maxima and minima, reduction principle, distance to reference copula, rearrangement algorithm.}

\keyAMSClassification{91B30, 62H05, 60E05, 60E15.}
\keyJELClassification{G32, C52, C60.}

\thanksColleagues{A previous version was entitled ``\textit{Model-free bounds on Value-at-Risk using partial dependence information}''.}

\date{} \maketitle \frenchspacing 

\section{Introduction}

The evaluation of multivariate risks under model uncertainty has become a central issue in several fields of science, ranging from hydrology and engineering to insurance and finance. 
In insurance and finance, this has been in parts driven by the changing regulations requiring the quantification of model uncertainty in risk management; see \textit{e.g.} the notions and guidelines regarding primary and secondary uncertainty of natural catastrophe models in the Swiss Solvency Test or the treatment of model uncertainty in the Solvency II guidelines for internal model approval. 
Measuring risk under uncertainty often relates to the computation of bounds on probabilities of the form $\mathbb{P}(\varphi(\mathbf X)\le \cdot)$, where $\mathbf X = (X_1,\dots,X_d)$ is an $\mathbb{R}^d$-valued random vector and $\varphi\colon\mathbb{R}^d\to\mathbb{R}$ an aggregation function. 
Here $\mathbf X$ can be thought of as a vector modeling $d$ risks in a portfolio, while typical examples of aggregation functions $\varphi$ are the sum, the max and the min operators. 

Models for the distribution of the risk factors $\mathbf X$ are exposed to two types of model risk, namely, the risk that the one-dimensional distributions of the individual constituents $X_1,\dots,X_d$ are misspecified and, on the other hand, the risk that the dependence structure between the components is not appropriate. 
The latter situation is referred to as \textit{dependence uncertainty} in the literature. 
While in many regulatory frameworks, the measurement of dependence uncertainty extends merely to the consideration of uncertain correlations, authorities are aware that the choice of the underlying dependence structure, \textit{i.e.} the copula, entails risks that are possibly far more significant than those ensuing from misspecified correlations; see \textit{e.g.} \citet{CEIOPS}. 
However, in view of a lack of parsimonious and numerically tractable methods to quantify risks due to dependence uncertainty, a standard framework in this respect seems currently impracticable. 

Against this backdrop, we focus, in this paper, on risk measurement under dependence uncertainty, \textit{i.e.} we assume that the marginal distributions of the components $X_i\sim F_i$ for $i=1,\dots,d$ are known, while the dependence structure between the components of $\mathbf{X}$ is unknown or only partially known. At first, we derive bounds on the distribution function of $\varphi(\mathbf X)$ using the available information on the distribution of $\mathbf X$. Then, by inversion, the bounds on the distribution of $\varphi(\mathbf X)$ can be translated immediately into bounds on the Value-at-Risk (VaR) of $\varphi(\mathbf X)$. 

A significant part of the literature focuses on the situation where only the marginals $F_1,\dots,F_d$ are known and no information on the dependence structure of $\mathbf{X}$ is available. 
In this case, explicit bounds on the distribution function of the sum of two random variables, \textit{i.e.} $\varphi(\mathbf{X}) = X_1+X_2$, were derived by \citet*{makarov1981} and for more general functions $\varphi$ by \citet*{rueschendorf1981} in the early 1980's. 
These results were later generalized for functions of more than two random variables, for instance by \citet*{denuit1999} for the sum and by \citet*{embrechts2003} and \citet*{embrechts2006} for more general aggregation functions; see also \citet{Cheung_Lo_2013}. 
These bounds however may fail to be sharp. 
Therefore, numerical schemes to compute sharp distributional bounds have become increasingly popular. 
The rearrangement algorithm, which was introduced by \citet*{puccetti2012} and \citet*{embrechts2013}, represents an efficient method to approximate sharp bounds on the VaR of the sum $X_1+\cdots+X_d$ under additional requirements on the marginal distributions $F_1,\dots,F_d$. 
Moreover, sharp analytical VaR bounds in the marginals-only case have been obtained in the literature under certain assumptions on the marginal distributions; see \textit{e.g.} \citet{rueschendorf1982}, \citet{embrechts2006}, \citet{puccetti2012b}, \citet{wang2013}, \citet{bernard2014} as well as the references therein. 
However, the complete absence of information on the dependence structure typically leads to very wide bounds that are not sufficiently informative for practical applications; see \textit{e.g.} \citet*{bernard2015}. 
Besides, a complete lack of information about the dependence structure of $\mathbf{X}$ is often unrealistic, since quantities such as correlations or the values of the distribution function of $\mathbf{X}$ at certain points can be estimated with a sufficient degree of accuracy. 
This calls for methods to account for additional information on the dependence structure in the computation of risk bounds. 

A variety of analytical and numerical approaches to derive risk bounds including additional dependence information have been recently developed. 
Analytical bounds were derived by \citet*{embrechts2003} and \citet*{embrechts2006} for the case that a lower bound on the copula of $\mathbf{X}$ is given. 
\citet*{rueschendorf1991}, \citet*{embrechts2010} and \citet*{puccetti2012b} established bounds when the laws of some lower dimensional marginals of $\mathbf{X}$ are known. Analytical bounds that account for positive or negative dependence assumptions were presented in \citet*{embrechts2003} and \citet*{rueschendorf2005}. 
\citet*{bernard2015b} derived risk bounds when an upper bound on the variance of $\varphi(\mathbf X)$ is prescribed, and presented a numerical scheme to efficiently compute these bounds. 
Moreover, numerical and analytical methods to obtain risk bounds in factor models were presented by \citet{bernard2016}, while \citet*{bernard2015} considered the case where the distribution of $\mathbf X$ is known only on a subset of its domain and established a version of the rearrangement algorithm to account for this type of dependence information. 
A detailed account of this literature appears in \citet*{rueschendorf2016}.

In this paper we develop alternative approaches to compute VaR bounds for aggregations of multiple risks in the presence of dependence uncertainty. 
After recalling several definitions and useful results in Subsection \ref{sec:copulas}, in Subsection \ref{boundsOnVar} we revisit the standard and improved standard bounds on VaR and provide a direct derivation of the improved standard bounds when $\varphi = \max$ or $\varphi = \min$.
In Section \ref{prescribedMax} we utilize a reduction principle to account for extreme value information, such as the distribution of partial minima or maxima of the risk vector $\mathbf X$, in the computation of risk bounds for the sum $X_1+\cdots+X_d$. 
The term partial maxima hereby refers to the maximum of lower dimensional marginals of $\mathbf X$, \textit{i.e.} $\max\{X_{i_1},\dots,X_{i_n}\}$ for $1\leq i_1\leq\cdots\leq i_n\leq d$, and analogously for the minimum. 
We thereby interpolate between the marginals-only case and the situation where the distributions of the lower-dimensional marginals of $\mathbf X$ are completely specified; \textit{cf.} \citep{embrechts2010, puccetti2012b}.

In Section \ref{boundsOnCopula} we present an approach to compute VaR bounds for general aggregation functions $\varphi$ including two different types of dependence information. 
First, we consider the situation where the copula $C$ of the risk vector $\mathbf X$ coincides with a reference model on a subset $\set$ of its domain, \textit{i.e.} it holds that $C(\bx) = C^*(\bx)$ for all $\bx\in\set$ and a reference copula $C^*$. 
Applying results from \citet{lux2016} and the improved standard bounds of \citet{embrechts2003} and \citet{embrechts2006} we derive bounds on VaR using the available information on the subset $\set$. 
This relates to the \textit{trusted region} in \citet*{bernard2015}, although the methods are different. 
The second type of dependence information corresponds to $C$ lying in the vicinity of a reference copula $C^*$ as measured by a statistical distance $\cD$. 
In this case we establish improved \FH bounds on the set of all (quasi-)copulas $C$ in the $\delta$-neighborhood of the reference model $C^*$, \textit{i.e.} for all $C$ such that $\cD(C,C^*)\leq\delta$. 
Our method applies to a large class of statistical distances such as the Kolmogorov--Smirnov or the Cram\'er--von Mises distances. 
We then use the improved standard bounds of \citep{embrechts2003, embrechts2006} in order to translate the improved \FH bounds into bounds on the VaR of $\varphi(\mathbf X)$. 

Finally, in Section \ref{numerics} we present several applications of our results in risk measurement. 
The computational results show that the additional dependence information typically leads to a significant improvement of the VaR bounds when compared to the marginals-only case.
Moreover, the VaR bounds using information on the partial maxima are becoming tighter as the confidence level increases, which is in contrast to related results in the literature, and constitutes an advantage of this methodology.

\section{Preliminaries}
\label{setting}

In this section we introduce the notation and some basic results that are used throughout this work. 
A comprehensive introduction to copulas in the context of risk aggregation can be found, for example, in \citet*{embrechts2015} or \citet*{rueschendorf2015}.

Let $d\geq2$ be an integer. In the following, $\mathbb{I}$ denotes the unit interval $[0,1]$, while boldface letters, \textit{e.g.} $\mathbf{u}$, $\mathbf{v}$ or $\mathbf{x}$, denote vectors in $\mathbb{I}^d$ or $\mathbb{R}^d$. Moreover, $\mathbf 1$ denotes the $d$-dimensional vector with all entries equal to one, \textit{i.e.} $\mathbf 1=(1,\dots,1)$.

\subsection{Copulas and \FH bounds}
\label{sec:copulas}

\begin{definition}
A function $Q\colon\mathbb{I}^d\to\mathbb{I}$ is a $d$-\textit{quasi-copula} if the following properties hold:
\begin{enumerate}[label={$(\mathbf{QC1})$},leftmargin=!,labelwidth=\widthof{\bfseries XXXX}]
\item $Q$ satisfies, for all $i\in\{1,\dots,d\}$, the boundary conditions 
      \label{cond:QC1}
      $$Q(u_1,\dots,u_i = 0,\dots,u_d)=0 
      \quad\text{ and }\quad
      Q(1,\dots,1,u_i,1,\dots,1) = u_i.$$ 
\end{enumerate}
\begin{enumerate}[label={$(\mathbf{QC2})$},leftmargin=!,labelwidth=\widthof{\bfseries XXXX}]
\item $Q$ is non-decreasing in each argument. \label{cond:QC2}
\end{enumerate}
\begin{enumerate}[label={$(\mathbf{QC3})$},leftmargin=!,labelwidth=\widthof{\bfseries XXXX}]
\item $Q$ is Lipschitz continuous, \textit{i.e.} for all 
      $\mathbf{u},\mathbf{v}\in\mathbb{I}^d$\label{cond:QC3}
      $$|Q(u_1,\dots,u_d)-Q(v_1,\dots,v_d)|\leq\sum_{i=1}^d |u_i-v_i|.$$ 
\end{enumerate}
Moreover, $Q$ is a $d$-\textit{copula} if
\begin{enumerate}[label={$(\mathbf{QC4})$},leftmargin=!,labelwidth=\widthof{\bfseries XXXX}]
\item $Q$ is $d$-increasing. \label{cond:QC4}
\end{enumerate}
\end{definition}

We denote the set of all $d$-quasi-copulas by $\mathcal{Q}^d$ and the set of all $d$-copulas by $\mathcal{C}^d$. 
Obviously $\mathcal{C}^d\subset\mathcal{Q}^d$. 
We will simply refer to a $d$-(quasi-)copula as a (quasi-)copula if the dimension is clear from the context. 

Let $C$ be a $d$-copula and consider $d$ univariate probability distribution functions $F_1,\dots,F_d$. 
Then $F(x_1,\dots,x_d):=C(F_1(x_1),\dots,F_d(x_d))$, for all $\mathbf{x}\in\mathbb{R}^d$, defines a $d$-dimensional distribution function with univariate margins $F_1,\dots,F_d$. 
The converse also holds by Sklar's Theorem, \textit{i.e.} for each $d$-dimensional distribution function $F$ with univariate marginals $F_1,\dots,F_d$, there exists a copula $C$ such that $F(x_1,\dots,x_d) = C(F_1(x_1),\dots,F_d(x_d))$ for all $\mathbf{x}\in\mathbb{R}^d$; see \citet{sklar1959}. 
In this case, the copula $C$ is unique if the marginals are continuous. 

The \textit{survival function} of a $d$-copula $C$ is defined as follows:
$$\widehat{C}(u_1,\dots,u_d) 
  := V_C([u_1,1]\times\cdots\times[u_d,1]),\quad \mathbf{u} \in \mathbb{I}^d,$$
where $V_C(H)$ denotes the $C$-volume of the set $H$.  
The function $\widehat{C}(\mathbf 1-\mathbf u)$, for $\mathbf{u} \in \mathbb{I}^d$, is again a copula, namely the \textit{survival copula} of $C$; see \textit{e.g.} \citet*{georges2001}.
Note that for a distribution function $F$ of a random vector $(X_1,\dots,X_d)$ with marginals $F_1,\dots,F_d$ and a corresponding copula $C$ such that $F(x_1,\dots,x_d) = C(F_1(x_1),\dots,F_d(x_d))$ it holds that 
\begin{align}
\label{survivalCopulaProbability}
\mathbb{P}(X_1>x_1,\dots,X_d>x_d) = \widehat{C}(F_1(x_1),\dots,F_d(x_d)).
\end{align}
The map $\widehat{Q}$ could be defined analogously for quasi-copulas $Q$, however the function $\widehat{Q}(\mathbf 1-\mathbf u)$ is not necessarily a quasi-copula again. 
Therefore, we introduce the term \textit{quasi-survival functions} to refer to functions $\widehat{Q}:\mathbb{I}^d\to\mathbb{I}$ such that $\mathbf u\mapsto\widehat{Q}(\mathbf 1-\mathbf u)$ is again a quasi-copula. 
The set of $d$-quasi-survival functions is denoted by $\widehat{\mathcal{Q}}^d$.

\begin{definition}
Let $Q,Q'$ be $d$-quasi-copulas. 
$Q'$ is greater than $Q$ in the \textit{lower orthant order}, denoted by $Q \preceq Q'$, if $Q(\mathbf{u})\leq Q'(\mathbf{u})$ for all $\mathbf{u}\in\mathbb{I}^d$. 
\end{definition}

The well-known Fr\'{e}chet--Hoeffding theorem establishes the minimal and maximal bounds on the set of quasi-copulas in the lower orthant order. 
In particular, for each $Q\in\mathcal{Q}^d$, it holds that
$$W_d(\mathbf{u})
  := \max\Big\{0,\sum_{i=1}^d u_i - d + 1\Big\} 
  \leq   Q(\mathbf{u}) \leq \min\{u_1,\dots,u_d\} 
  =: M_d(\mathbf{u}),$$
for all $\mathbf{u}\in\mathbb{I}^d$, \textit{i.e.} $W_d\preceq Q\preceq M_d$, where $W_d$ and $M_d$ are the lower and upper Fr\'{e}chet--Hoeffding bounds respectively. 
The properties of the \FH bounds carry over to the set of survival copulas in a straightforward way, hence one obtains similarly for any $C\in\mathcal{C}^d$ the following bounds:
$$W_d(\mathbf 1 - \bu) \le \widehat{C}(\bu) \le
  M_d(\mathbf 1 - \bu), \qquad\text{for all }\mathbf{u}\in\mathbb{I}^d.$$

\subsection{Bounds on Value-at-Risk}
\label{boundsOnVar}

The problem of computing bounds on the probability of a function of random variables, or equivalently on their Value-at-Risk, in the presence of dependence uncertainty has a long history and many approaches to its solution have emerged. 
In the situation of complete dependence uncertainty, where only the marginals $F_1,\dots,F_d$ are known and one has no information about the copula of $\mathbf{X}$, bounds for the quantiles of the sum $X_1+\cdots+X_d$ were derived in a series of papers, starting with the results by \citet{makarov1981} and \citet{rueschendorf1982} for $d=2$, and their extensions for $d>2$ by \citet{frank1987}, \citet{denuit1999} and \citet{embrechts2003}. 
These bounds are in the literature referred to as \textit{standard bounds} and they are given by
\begin{align}\label{standardBounds}
\begin{split}
\max\Big\{\sup_{\mathcal{U}(s)}\Big(F^-_1(u_1)+\sum_{i=2}^d F_i(u_i)\Big)-d+1,0\Big\} 
	&\leq \P(X_1+\cdots+X_d<s)\\ 
	&\qquad \leq\min\Big\{\inf_{\mathcal{U}(s)}\sum_{i=1}^d F^-_i(u_i),1\Big\},
\end{split}
\end{align}
where $\mathcal{U}(s) = \{(u_1,\dots,u_d)\in\mathbb{R}^d\colon u_1+\cdots+u_d=s\}$ and $F_i^-$ denotes the left-continuous version of $F_i$. 
These bounds hold for all random variables $\mathbf{X}$ with marginals $F_1,\dots,F_d$, and the corresponding bounds for the VaR of the sum $X_1+\cdots+X_d$ are given by the respective inverse functions. 
It was shown independently in \citep{makarov1981} and \citep{rueschendorf1982} that the bounds are sharp for $d=2$, in the sense that there exists a distribution for $\mathbf{X}$ such that the sum of its constituents attains the upper and the lower bound. 
The standard bounds may however fail to be sharp in higher dimensions. 

\citet{embrechts2003} and \citet{embrechts2006} derived an improvement of the standard bounds that accounts for a lower bound on the copula of $\mathbf{X}$ or its survival function. 
This improvement is essential for the results in the present work, since it relates the problem of computing improved VaR bounds in the presence of additional dependence information to the task of improving the \FH bounds on copulas. 
The improvement of the `classical' \FH bounds by using additional, partial information on the dependence structure has attracted some attention in the literature lately, see \textit{e.g.} \citet{nelsen2006}, \citet{tankov2011}, \citet{lux2016} and also \citet{Rachev_Rueschendorf_1994}.

Let $\X$ be a random vector with marginals $F_1,\dots,F_d$ and copula $C$, let $\varphi\colon\mathbb{R}^d\to\mathbb{R}$ be non-decreasing in each coordinate, and define the functional
\[
\P_C(\varphi(\mathbf{X})<s) := \int_{\mathbb{R}^d} \mathds{1}_{\{\varphi(x_1,\dots,x_d)<s\}}\ \ud C(F_1(x_1),\dots,F_d(x_d)).
\]
Let $C_0,C_1$ be copulas and consider the following quantities
\begin{align*}
m_{C_0,\varphi}(s) 
	&:= \inf\big\{\P_C(\varphi(\mathbf{X})<s) \colon C\in\mathcal C^d, C_0\preceq C\big\}, \\
M_{\widehat C_1,\varphi}(s) 
	&:= \sup\big\{\P_C(\varphi(\mathbf{X})<s) \colon C\in\mathcal C^d, \widehat{C}_1\preceq \widehat{C}\big\}.
\end{align*}
The following bounds on $m_{C_0,\varphi}, M_{\widehat C_1,\varphi}$ are known in the literature as \textit{improved standard bounds} and read as follows:
\begin{align}\label{eq:ISB}
\begin{split}
m_{C_0,\varphi}(s) 
	& \geq \sup_{\mathcal V^<_\varphi(s)}\ C_0\big(F_1(x_1),\dots,F_{d}(x_{d})\big) 
	=: \underline{m}_{C_0,\varphi}(s),\\ 
M_{\widehat C_1,\varphi}(s) 
	& \leq \inf_{\mathcal V^>_\varphi(s)}\ 1-\widehat{C}_1\big(F_1(x_1),\dots,F_{d}(x_{d})\big) 
	=: \overline{M}_{\widehat C_1,\varphi}(s),
\end{split}
\end{align}
where $\mathcal V^<_\varphi(s) = \{ (x_1,\dots,x_d)\in\R^d: \varphi(\bx)< s \}$ and $\mathcal V^>_\varphi(s) = \{ (x_1,\dots,x_d)\in\R^d: \varphi(\bx)> s \}$; see \cite{embrechts2003,embrechts2006}.
A careful examination of the proof of Theorem 3.1 in \citet{embrechts2006} reveals that these results hold also when $C_0$, resp. $\widehat{C}_1$, is just increasing, resp. decreasing, in each coordinate. 
Hence, they hold in particular when $C_0$ is a quasi-copula and $\widehat{C}_1$ a quasi-survival function. 
The above bounds relate to the VaR of $\varphi(\mathbf{X})$ in the following way.

\begin{remark}\label{varBoundsRemark}
Let $\varphi\colon\mathbb{R}^d\to\mathbb{R}$ be increasing in each component and the copula $C$ of $\mathbf{X}$ be such that $Q_0\preceq C$ and $\widehat{Q}_1\preceq\widehat{C}$, for a quasi-copula $Q_0$ and a quasi-survival function $\widehat{Q}_1$. 
Then
$$
\overline{M}^{-1}_{\widehat Q_1,\varphi}(\alpha) 
	\leq \mathrm{VaR}_\alpha(\varphi(\mathbf{X})) 
	\leq \underline{m}^{-1}_{Q_0,\varphi}(\alpha).
$$
\end{remark}

Besides the aggregation function $\varphi(x_1,\dots,x_d) = x_1+\cdots+x_d$, the operations $\varphi(x_1,\dots,x_d) = \max\{x_1,\dots,x_d\}$ and $\varphi(x_1,\dots,x_d) = \min\{x_1,\dots,x_d\}$ are also of particular interest in risk management, however fewer methods to handle dependence uncertainty for these operations exist; \textit{cf.} \citet{embrechts2014}. 
The following result, whose proof is deferred to Appendix \ref{app:easy-proof}, establishes bounds for the minimum and maximum operations in the presence of additional information on the copula using straightforward computations, and further shows that these bounds coincide with the improved standard bounds \eqref{eq:ISB}. 
Analogous statements for $d=2$ in the absence of additional information on the copula $C$ can be found in \citet[Theorem~5.1]{frank1987}.

\begin{proposition}
\label{varBoundsMax}
Let $\mathbf X$ be a random vector with copula $C$ and marginals $F_1,\dots,F_d$, and let $\underline{Q},\overline{Q}$ be quasi-copulas. 
Then, for $\varphi(x_1,\dots,x_d)=\max\{x_1,\dots,x_d\}$, we have that 
\begin{align*}
m_{\underline{Q},\max}(s) 
	&= \inf\big\{\P_C(\varphi(\mathbf{X})<s)\colon \underline{Q}\preceq C\big\} 
	\geq \underline{Q}(F_1(s),\dots,F_d(s)) 
	=: \underline{m}_{\underline{Q},\max}(s)\\
M_{\overline{Q},\max}(s) 
	&= \sup\big\{\P_C(\varphi(\mathbf{X})<s)\colon C\preceq \overline{Q}\big\} 
	\leq \overline{Q}(F_1(s),\dots,F_d(s)).
\end{align*}
Analogously, if $\underline{\widehat{Q}}$ and $\widehat{\overline{Q}}$ are quasi-survival functions then, for $\varphi(x_1,\dots,x_d)=\min\{x_1,\dots,x_d\}$, we have that  
\begin{align*}
m_{\widehat{\overline{Q}},\min}(s) 
	&= \inf\big\{\P_C(\varphi(\mathbf{X})<s)\colon \widehat{C}\preceq \widehat{\overline{Q}}\big\} 
	\geq 1-\widehat{\overline{Q}}(F_1(s),\dots,F_d(s))\\
M_{\widehat{\underline{Q}},\min}(s) 
	&= \sup\big\{\P_C(\varphi(\mathbf{X})<s)\colon \widehat{\underline{Q}}\preceq \widehat{C}\big\} 
	\leq 1-\widehat{\underline{Q}}(F_1(s),\dots,F_d(s)) 
	=: \overline{M}_{\widehat{\underline{Q}},\min}(s).
\end{align*}
\end{proposition}

\section{Improved bounds on the Value-at-Risk of the sum with known distributions of some minima or maxima}
\label{prescribedMax}

In this section, we provide improved bounds on the VaR of the sum $X_1+\cdots+X_d$ in the situation where, besides the marginal distributions, the laws of the minima and maxima of some subsets of the risks $X_1,\dots,X_d$ are known. 
In particular, we assume that for a system $J_1,\dots,J_m\subset\{1,\dots,d\}$ the distribution of $\max_{j\in J_n} X_j$ or $\min_{j\in J_n} X_j$ for $n=1,\dots,m$ is given. 
This setting can be viewed as an interpolation between the marginals-only case and the situation where the lower-dimensional marginals of the vectors $(X_j)_{j\in J_n}$ are completely specified. 
The latter setting has been studied extensively in the literature, and risk bounds for aggregations of $(X_1,\dots,X_d)$ given some of its lower-dimensional marginals were obtained, for instance, by \citet{rueschendorf1991}, \citet{embrechts2010} and \citet{puccetti2012b}. 
These bounds are based on a reduction principle that transforms the optimization problem involving higher-dimensional marginals into a standard Fr\'echet problem (\textit{i.e.} marginals-only), utilizing the extra information about the distribution of the subvector $(X_j)_{j\in J_n}$.
In practice however it is often difficult to determine the distributions of the lower-dimensional vectors $(X_j)_{j\in J_n}$. 
In particular when the subsets $J_1,\dots,J_m$ are high dimensional, a vast amount of data is required in order to estimate the distribution of $(X_j)_{j\in J_n}$ with an adequate degree of accuracy. 
Therefore, having complete information about lower-dimensional marginals of $(X_1,\dots,X_d)$ turns out to be a rather strong assumption, while methods that interpolate between this scenario and the marginals-only case are of practical interest. 
Based on the reduction principle of \citep{puccetti2012b}, we develop in this section improved bounds on the VaR of the sum $X_1+\cdots+X_d$ when, instead of the distribution of $(X_j)_{j\in J_n}$, only the distribution of its maximum $\max_{j\in J_n} X_j$ or minimum $\min_{j\in J_n} X_j$ is known. 

\begin{remark}
Another work in the spirit of interpolation between the marginals-only case and the case of full knowledge of the lower-dimensional marginals is \citet{rueschendorf2017}, where knowledge of dependence information on subgroups of $(X_1,\dots,X_d)$ is assumed. 
\end{remark}

Let us denote by $\mathcal{I}:=\{1,\dots,d\}$ and by $\mathcal{J}:=\{1,\dots,m\}$.

\begin{theorem}\label{boundMax}
Let $(X_1,\dots,X_d)$ be a random vector with marginals $F_1,\dots,F_d$, and consider a collection $\mathcal{E} = \{J_1,\dots,J_m\}$ of subsets $J_n\subset\mathcal{I}$ for $n\in\mathcal{J}$ with $\bigcup_{n\in\mathcal{J}} J_{n} = \mathcal{I}$. 
Denote by $G_n$ the distribution of $Y_n=\max_{j\in J_n} X_j$. 
Then it follows that
\begin{multline*}
\inf\big\{\P(X_1+\cdots+X_d\leq s)\colon X_i\sim F_i, i\in\mathcal{I}, \max_{j\in J_n} X_j \sim G_n, n\in\mathcal{J}\big\}\\
	\geq \sup_{(\alpha_1,\dots,\alpha_m)\in\underline{\mathcal{A}}} \inf\big\{\P(\alpha_1Y_1+\cdots+\alpha_mY_m\leq s)\colon Y_n\sim G_n, n \in\mathcal{J}\big\} 
	=: \underline{m}_{\mathcal{E},\max}(s),
\end{multline*}
where 
$$\underline{\mathcal{A}} = \Big\{(\alpha_1,\dots,\alpha_m)\in\mathbb{R}_+^m\colon \sum_{n=1}^m \alpha_n\max_{j\in J_n} x_j \geq \sum_{i=1}^d x_i,\text{ for all } (x_1,\dots,x_d)\in\mathbb{R}^d\Big\}\neq\emptyset.$$ 
Moreover if $(X_1,\dots,X_d)$ is $\mathbb{R}_+^d$-valued, then
\begin{multline*}
\sup\big\{\P(X_1+\cdots+X_d\leq s)\colon X_i\sim F_i, i\in\mathcal{I}, \max_{j\in J_n} X_j \sim G_n, n\in\mathcal{J}\big\}\\
	\leq  \inf_{(\alpha_1,\dots,\alpha_m)\in\overline{\mathcal{A}}}  \sup\big\{\P(\alpha_1Y_1+\cdots+\alpha_mY_m\leq s)\colon Y_n\sim G_n, 	n\in\mathcal{J}\big\} 
	=: \overline{M}_{\mathcal{E},\max}(s),
\end{multline*}
where 
$$\overline{\mathcal{A}} = \Big\{(\alpha_1,\dots,\alpha_m)\in\mathbb{R}_+^m\colon \sum_{n=1}^m \alpha_n\max_{j\in J_n} x_j \leq \sum_{i=1}^d x_i,\text{ for all } (x_1,\dots,x_d)\in\mathbb{R}_+^d\Big\}\neq\emptyset.$$
\end{theorem}

\begin{proof}
We first show that the lower bound $\underline{m}_{\mathcal{E},\max}$ is valid. 
It follows from $\bigcup_{n=1}^m J_n = \{1,\dots,d\}$ that $\underline{\mathcal{A}}\neq\emptyset$. 
Indeed, choosing for instance $\alpha_n=|J_n|$ we get that $\sum_{j\in J_n} x_j \leq \alpha_n \max_{j\in J_n} x_j$, for all $(x_1,\dots,x_d)\in\mathbb{R}^d$ and $n=1,\dots,m$.
Hence
$$\sum_{n=1}^m \alpha_n\max_{j\in J_n} x_j
	\geq \sum_{n=1}^m\sum_{j\in J_n} x_j 
	\geq \sum_{i=1}^d x_i\quad\text{for all }(x_1,\dots,x_d)\in\mathbb{R}^d.$$
Then, it follows for arbitrary $(\alpha_1,\dots,\alpha_m)\in\underline{\mathcal{A}}$ that
$$\bigg\{\sum_{n=1}^m \alpha_n\max_{j\in J_n} X_j\leq s\bigg\} 
	\subseteq \bigg\{\sum_{i=1}^d X_i\leq s\bigg\},$$
henceforth
\begin{align*}
\inf & \big\{\P(X_1+\cdots+X_d\leq s)\colon X_i\sim F_i, i\in\mathcal{I}, \max_{j\in J_n} X_j \sim G_n, n\in\mathcal{J}\big\} \\
	&\geq \inf\bigg\{\P\bigg(\sum_{n=1}^m \alpha_n\max_{j\in J_n} X_j\leq s\bigg)\colon X_i\sim F_i, i\in\mathcal{I}, \max_{j\in J_n} X_j \sim 		G_n, n\in\mathcal{J}\bigg\}\\
	&= \inf\big\{\P(\alpha_1Y_1+\cdots+\alpha_mY_m\leq s)\colon Y_n\sim G_n, n\in\mathcal{J}\big\}.
\end{align*} 
Now, since $(\alpha_1,\dots,\alpha_m)\in\underline{\mathcal{A}}$ was arbitrary, it follows that the lower bound holds by taking the supremum over all elements in $\underline{\mathcal{A}}$.

Likewise for the upper bound, we note that since $(X_1,\dots,X_d)$ is $\mathbb{R}^d_+$-valued, the vectors $(0,\dots,0)$ and $(1,\dots,1)$ belong to $\overline{\mathcal{A}}$, hence it is not empty. 
Moreover, for arbitrary $(\alpha_1,\dots,\alpha_m)\in\overline{\mathcal{A}}$, it follows that 
$$
\bigg\{\sum_{n=1}^m \alpha_n\max_{j\in J_n} X_j\leq s\bigg\} 
	\supseteq \bigg\{\sum_{i=1}^d X_i\leq s\bigg\},
$$
due to the fact that $(X_1,\dots,X_d)$ is non-negative and $\sum_{j\in J_n} x_j \geq \sum_{n=1}^m \alpha_n\max_{j\in J_n} x_j$. 
Hence, we get that 
\begin{align*}
\sup\big\{ & \P(X_1+\cdots+X_d\leq s)\colon X_i\sim F_i, i\in\mathcal{I}, \max_{j\in J_n} X_j \sim G_n, n\in\mathcal{J}\big\} \\
	&\leq\sup\bigg\{\P\bigg(\sum_{n=1}^m \alpha_n\max_{j\in J_n} X_j\leq s\bigg)\colon X_i\sim F_i, i\in\mathcal{I}, \max_{j\in J_n} X_j \sim G_n, n\in\mathcal{J}\bigg\}\\
	&=\sup\big\{\P(\alpha_1 Y_1+\cdots+\alpha_m Y_m\leq s)\colon Y_n\sim G_n, n\in\mathcal{J}\big\}. 
\end{align*}
Since $(\alpha_1,\dots,\alpha_m)\in\overline{\mathcal{A}}$ was arbitrary, it follows that the upper bound holds indeed. 
\end{proof}

\begin{remark}\label{rem:calE}
The assumption $\bigcup_{n=1}^m J_{n} = \{1,\dots,d\}$ can always be met by adding singletons to $\mathcal{E}$, \textit{i.e.} $J_n = \{i_n\}$ for $i_n\in\{1,\dots,d\}$, since the marginal distributions of $(X_1,\dots,X_d)$ are known. 
However, the bounds are valid even when the marginal distributions are not known.
\end{remark}

By the same token, the following result establishes bounds on the distribution of the sum of the components of $\mathbf{X}$ when distributions of some minima are known.
The proof follows along the same lines of argumentation as the proof of Theorem \ref{boundMax}, and is therefore omitted.

\begin{theorem}\label{boundMin}
Consider the setting of Theorem \ref{boundMax} and denote by $H_n$ the distribution of $Z_n=\min_{j\in J_n} X_j$. 
Then it follows that 
\begin{multline*}
\sup\big\{\P(X_1+\cdots+X_d\leq s)\colon X_i\sim F_i, i\in\mathcal{I}, \min_{j\in J_n} X_j \sim H_n, n\in\mathcal{J}\big\}\\
	\leq \inf_{(\alpha_1,\dots,\alpha_m)\in\overline{\mathcal{B}}} \sup\big\{\P(\alpha_1Z_1+\cdots+\alpha_mZ_m\leq s)\colon Z_n\sim H_n, n\in\mathcal{J}\big\} 
	=: \overline{M}_{\mathcal{E},\min}(s),
\end{multline*}
where 
$$\overline{\mathcal{B}} = \Big\{(\alpha_1,\dots,\alpha_m)\in\mathbb{R}_+^m\colon \sum_{n=1}^m \alpha_n\min_{j\in J_n} x_j \leq \sum_{i=1}^d x_i,\text{ for all } (x_1,\dots,x_d)\in\mathbb{R}^d\Big\}\neq\emptyset.$$
Moreover if $(X_1,\dots,X_d)$ is $\mathbb{R}_-^d$-valued, then 
\begin{multline*}
\inf\big\{\P(X_1+\cdots+X_d\leq s)\colon X_i\sim F_i, i\in\mathcal{I}, \min_{j\in J_n} X_j \sim H_n, n\in\mathcal{J}\big\}\\
	\geq \sup_{(\alpha_1,\dots,\alpha_m)\in\underline{\mathcal{B}}}\inf\big\{\P(\alpha_1 Z_1+\cdots+\alpha_m Z_m\leq s)\colon Z_n\sim H_n, n\in\mathcal{J}\big\} 
	=: \underline{m}_{\mathcal{E,\min}}(s),
\end{multline*}
where 
$$\underline{\mathcal{B}} = \Big\{(\alpha_1,\dots,\alpha_m)\in\mathbb{R}_+^m\colon \sum_{n=1}^m \alpha_n\min_{j\in J_n} x_j \geq \sum_{i=1}^d x_i,\text{ for all } (x_1,\dots,x_d)\in\mathbb{R}_-^d\Big\}\neq\emptyset.$$
\end{theorem}

The computation of the bounds presented in Theorems \ref{boundMax} and \ref{boundMin} can be cumbersome for two reasons. 
Firstly, for fixed  $(\alpha_1,\dots,\alpha_m)$ there does not exist a method to compute sharp analytical bounds on the set $\big\{\P(\alpha_1Y_1+\cdots+\alpha_mY_m\leq s)\colon Y_n\sim G_n, n = 1,\dots,m\}$, except when $m=2$. 
This problem can be circumvented either by using the standard bounds in \eqref{standardBounds}, or numerically, by an application of the rearrangement algorithm of \citet{embrechts2013}; see Appendix \ref{sec:reduction-arg} for more details. 
Using the rearrangement algorithm, we are able to approximate upper and lower bounds on the set in an efficient way. In Section \ref{numerics} we demonstrate, that the bounds in Theorems \ref{boundMax} and \ref{boundMin} yield a significant improvement over sharp bounds available in the literature, that assume only knowledge of the marginals.

\section{Improved \FH bounds on copulas using a subset or a reference copula}
\label{boundsOnCopula}

A general method to derive sharper bounds on the Value-at-Risk of $\varphi(\mathbf X)$, for general aggregation functions $\varphi$, is to first derive improved \FH bounds on the copula of $\mathbf X$ by assuming that additional dependence information is available, and then to translate them into VaR bounds using the improved standard bounds \eqref{eq:ISB}; see also Remark \ref{varBoundsRemark}.
In this section, we focus on the first part of this strategy and discuss improved \FH bounds using two types of additional dependence information.

Firstly, we consider the situation where the copula $C$ of the risk vector $\mathbf X$ coincides with a reference model on a compact subset $\set$ of its domain, \textit{i.e.} it holds that $C(\bx) = C^*(\bx)$ for all $\bx\in\set$ and a reference copula $C^*$. 
In practice, the set $\set$ may correspond to a region in $\mathbb{I}^d$ that contains enough observations to estimate the copula $C$ with sufficient accuracy, so that we can assume that $C$ is known on $\set$. 
\citet{bernard2015} call such a subset \textit{trusted region} and present several techniques and criteria to select such regions when estimating copulas. 
If $\set$ is not equal to the entire domain of the copula, then dependence uncertainty stems from the fact that $C$ remains unknown on $\mathbb{I}^d\setminus\set$. 
In order to obtain VaR bounds in this situation, we use results from \citet{lux2016} who established improved \FH bounds on the set of copulas with prescribed values on a compact set. 

Secondly, we present a new improvement of the \FH bounds when the copula $C$ is assumed to lie in the vicinity of a reference model as measured by a statistical distance. 
More formally, we establish bounds on the set of all (quasi-)copulas $C$ in the $\delta$-neighborhood of the reference copula $C^*$, \textit{i.e.} such that $\cD(C,C^*)\leq\delta$ for a distance $\cD$. 
Our method applies to a large class of statistical distances such as the Cram\'er--von Mises or the $L^p$ distances. 
Such situations arise naturally in practice when one tries to estimate a copula from, or calibrate it to, empirical data. 
The estimation typically involves the minimization of a distance to the empirical copula over a parametric family of copulas, \textit{i.e.} $\cD(C_\theta,C^*)\to \min_\theta$ where $C^*$ is an empirical copula and $(C_\theta)_\theta$  is a family of parametric copulas. 
This is in the literature often referred to as \textit{minimal distance} or \textit{minimal contrast} estimation. 
\citet{kole2007} for instance present several distance-based techniques for selecting copulas in risk management. 
These estimation procedures lend themselves immediately to the methodology we propose, as typically one arrives at $\delta:=\min_\theta\cD(C_\theta,C^*)>0$, due to the fact that the family of models $(C_\theta)_\theta$ is not able to match the empirical observations exactly, thus dependence uncertainty remains. 
In this case, $\delta$ can be viewed as the inevitable model risk due to the choice of the parametric family $(C_\theta)_\theta$. 
Our method can then be used to account for such types of dependence uncertainty in the computation of VaR. 

Approaches to compute robust risk estimates over a class of models that lie in the proximity of a reference model have been proposed earlier in the literature. 
\citet{glasserman2013} derive robust bounds on the portfolio variance, the conditional VaR and the CVA over the class of models within a relative entropy distance of a reference model. 
\citet{barrieu2015} establish bounds on the VaR of a univariate random variable given that its distribution is close to a reference distribution in the sense of the Kolmogorov--Smirnov or L\'evy distance. 
In a multivariate setting, \citet{blanchet2016} use an optimal transport approach to derive robust bounds on risk estimates, such as ruin probabilities, over models that are in a neighborhood of a reference model in terms of the Wasserstein distance. 
This brief overview is, of course, incomplete and we refer the reader to the references in each of the aforementioned articles for a more detailed review of the associated literature.

\subsection{Improved \FH bounds using subsets}

Let us consider the setting where, apart from the marginal distributions, partial information on the dependence structure of the random vector $\mathbf X$ is available.
In particular, assume that the copula is known on some subset $\mathcal S$ of $[0,1]^d$. 
Theorem 3.1 in \citep{lux2016} establishes sharp bounds on the set 
\begin{align*}
\mathcal{Q}^{\set,Q^*} 
	:= \big\{Q\in\mathcal{Q}^d\colon Q(\mathbf{x}) = Q^*(\mathbf{x}) \text{ for all } \mathbf{x}\in \set\big\},
\end{align*}
where $\set\subset\mathbb{I}^d$ is compact and $Q^*$ is a $d$-quasi-copula. 
The bounds are provided by
\begin{align}
\label{bounds}
\begin{split}
\underline{Q}^{\set,Q^*}(\mathbf{u}) 
	:=& \min\big\{Q(\bu)\colon Q(\bx) = Q^*(\bx) \text{ for all } \bx\in\set\big\}\\
  	 =& \max\Big(0, \sum_{i=1}^d u_i-d+1,\max_{\mathbf{x}\in \set} 
    	\Big\{Q^*(\mathbf{x})-\sum_{i=1}^d (x_i-u_i)^+\Big\}\Big),\\
\overline{Q}^{S,Q^*}(\mathbf{u}) 
	:=& \max\big\{Q(\bu)\colon Q(\bx) = Q^*(\bu) \text{ for all } \bx\in\set\big\}\\
  	 =& \min\Big(u_1,\dots,u_d,\min_{\mathbf{x}\in S} \Big\{Q^*(\mathbf{x})+\sum_{i=1}^d (u_i-x_i)^+\Big\}\Big),
\end{split}
\end{align}
for all $\bu\in\mathbb{I}^d$, they are quasi-copulas, and also belong to $\mathcal{Q}^{S,Q^*}$.
Let us point out that a similar version of these bounds was presented recently by \citet{puccetti2016}. 
They were derived independently in the master thesis of the third-named author. 

\begin{remark}
By slightly abusing notation, we will sometimes write $\underline{Q}^{\{\bu\},\alpha}$ and $\overline{Q}^{\{\bu\},\alpha}$ with $\alpha\in [W_d(\bu),M_d(\bu)]$ instead of a quasi-copula function $Q^*$, and mean that $Q^*(\bu)=\alpha$.
\end{remark}

The bounds in \eqref{bounds} hold also for sets of copulas, \textit{i.e.} for each copula $C$ in 
$$\mathcal{C}^{\set,Q^*} 
	:= \big\{C\in\mathcal{C}^d\colon C(\mathbf{x}) = Q^*(\mathbf{x}) \text{ for all } \mathbf{x}\in \set\big\}$$
it holds that $\underline{Q}^{\set,Q^*}\lo C\lo\overline{Q}^{\set,Q^*}$, assuming that $\mathcal{C}^{\set,Q^*}$ is not empty. 
Moreover, Proposition A.1 in \citep{lux2016} provides analogous bounds on survival functions, \textit{i.e.} for a reference copula $C^*$ and any copula $C$ in
$$\widehat{\mathcal{C}}^{\set,C^*} 
	:= \big\{C\in\mathcal{C}^d\colon \widehat{C}(\mathbf{x}) = \widehat{C}^*(\mathbf{x}) \text{ for all } \mathbf{x}\in\set\big\}$$
it holds that $\widehat{\underline{Q}}^{\set,C^*} \lo \widehat{C}\lo \widehat{\overline{Q}}^{\set,C^*}$, where
\begin{align}
\label{survivalBounds}
\widehat{\underline{Q}}^{\set,C^*}(\mathbf{u}) 
	:= \underline{Q}^{\widehat{\set},\widehat{C}^*}(\mathbf 1-\bu) 
		\quad\text{and}\quad 
\widehat{\overline{Q}}^{\set,C^*}(\mathbf{u}) 
	:= \overline{Q}^{\widehat{\set},\widehat{C}^*}(\mathbf 1-\bu),
\end{align}
while $\widehat{\set} = \{(1-x_1,\dots,1-x_d)\colon (x_1,\dots,x_d)\in \set\}$.

In case $d=2$, the above bounds correspond to the improved \FH bounds derived by \citet{tankov2011}. 
He showed that the bounds are themselves copulas under certain constraints on the set $\set$, and those were readily relaxed by \citet{bernard2012}. 
In contrast, \citet{lux2016} showed that for $d>2$ the bounds $\underline{Q}^{\set,Q^*}$ and $\overline{Q}^{\set,Q^*}$ are copulas only in degenerate cases, and quasi-copulas otherwise.  
Moreover, \citet{Bartl_Kupper_Lux_Papapantoleon_2017} recently showed that once the constraints of \cite{tankov2011,bernard2012} are violated then the improved \FH bounds fail to even be pointwise sharp, still in dimension $d=2$.

\subsection{Improved \FH bounds using a reference model}

In the following we will establish improved \FH bounds using a different type of additional dependence information. 
Namely, we consider the set of copulas that are close to a reference copula in the sense of a statistical distance as defined below.
Let us first define the minimal and maximal convolution between two quasi-copulas $Q,Q'$ as the pointwise minimum and maximum between them, \textit{i.e.} $(Q\wedge Q')(\bu) = Q(\bu) \wedge Q'(\bu)$ and $(Q\vee Q')(\bu) = Q(\bu) \vee Q'(\bu)$.

\begin{definition}
A function $\cD\colon\mathcal{Q}^d\times\mathcal{Q}^d\to\mathbb{R}_+$ is called a \emph{statistical distance} if for $Q,Q'\in\mathcal{Q}^d$
$$\cD(Q,Q') = 0\quad\Longleftrightarrow\quad Q(\bu)=Q'(\bu)\quad \text{for all }\bu\in\mathbb{I}^d.$$
\end{definition}

\begin{definition}
\label{statisticalDistance}
A statistical distance $\cD$ is \emph{monotonic} with respect to the order $\preceq$ on $\mathcal{Q}^d$, if for $Q,Q',Q''\in\mathcal{Q}^d$ it holds
\begin{align*}
Q\preceq Q'\preceq Q'' 
	\quad\Longrightarrow\quad \cD(Q',Q'')\leq\cD(Q,Q'') \ \text{ and } \ \cD(Q'',Q')\leq\cD(Q'',Q).
\end{align*}
A statistical distance $\cD$ is \emph{min-} resp. \emph{max-stable} if for $Q,Q'\in\mathcal{Q}^d$ it holds 
\begin{align*}
\cD(Q,Q') & \geq \max\{\cD({Q\wedge Q'},Q), \cD(Q,{Q\wedge Q'})\} \\
\cD(Q,Q') & \geq \max\{\cD({Q\vee Q'},Q), \cD(Q,{Q\vee Q'})\}.
\end{align*}
\end{definition}

The following theorem establishes pointwise bounds on the set of quasi-copulas that are in the $\delta$-vicinity of a reference copula $C^*$ as measured by a statistical distance $\cD$.
This result continues the line of research initiated by \citet{nelsen2006} and continued by \citet{tankov2011} and \citet{lux2016} on improved \FH bounds in case some dependence functional is known.

\begin{theorem}
\label{prescribedDistance}
Let $C^*$ be a $d$-copula and $\cD$ be a statistical distance which is continuous with respect to the pointwise convergence of quasi-copulas, monotonic with respect to the lower orthant order and min/max-stable. 
Consider the set 
$$\mathcal{Q}^{\cD,\delta} := \big\{ Q\in\mathcal{Q}^d\colon \cD(Q,C^*) \leq \delta \big\}$$
for $\delta\in\mathbb{R}_+$. 
Then
\begin{align*}
\underline{Q}^{\cD,\delta}(\bu) 
	&:=\min\Big\{\alpha \in \mathbb S(\bu) \colon \cD\Big({\upQ^{\{\bu\},\alpha}\wedge\refC},\refC\Big) \leq \delta\Big\}
	 	= \min\big\{Q(\bu)\colon Q\in\mathcal{Q}^{\cD,\delta}\big\},\\
\overline{Q}^{\cD,\delta}(\bu) 
	&:=\max\Big\{\alpha \in \mathbb S(\bu) \colon \cD\Big({\downQ^{\{\bu\},\alpha}\vee\refC},\refC\Big) \leq \delta\Big\}
		= \max\big\{Q(\bu)\colon Q\in\mathcal{Q}^{\cD,\delta}\big\},
\end{align*}
where $\mathbb S(\bu) := [W_d(\bu),M_d(\bu)]$, and both bounds are quasi-copulas. 
\end{theorem}

\begin{proof}
We show that the statement holds for the lower bound, while the proof for the upper bound follows along the same lines. 
Fix an $\alpha\in[W_d(\bu),M_d(\bu)]$ and a $\bu\in\mathbb I^d$, then the map $v\mapsto\big({\upQ^{\{\bu\},\alpha}\wedge\refC}\big)(v)$ is a quasi-copula; this follows by straightforward calculations using the definition of the minimal convolution, see also \citet[Theorem 2.1]{Rodriguez_Ubeda_2004}. 
By definition, $\cD$ is monotonic with respect to the lower orthant order, thus it follows for $\underline{\alpha},\overline{\alpha} \in [W_d(\bu),M_d(\bu)]$ with $\underline{\alpha}<\overline{\alpha}$ that
$$\cD\Big({\upQ^{\{\bu\},\overline{\alpha}}\wedge\refC},\refC\Big) 
	\leq \cD\Big({\upQ^{\{\bu\},\underline{\alpha}}\wedge\refC},\refC\Big),$$ 
due to the fact that $\upQ^{\{\bu\},\underline{\alpha}} \lo \upQ^{\{\bu\},\overline{\alpha}}$, which readily implies
$$\Big({\upQ^{\{\bu\},\underline{\alpha}}\wedge\refC}\Big)\lo \Big({\upQ^{\{\bu\},\overline{\alpha}}\wedge\refC}\Big) \lo \refC.$$
Hence, the map 
$$[W_d(\bu),M_d(\bu)]\ni\alpha\mapsto\cD\Big({\upQ^{\{\bu\},\alpha}\wedge\refC},\refC\Big)$$
is decreasing. 
Moreover, as a consequence of the Arzel\`{a}--Ascoli Theorem, it follows that for every sequence $(\alpha_n)_n\subset[W_d(\bu),M_d(\bu)]$ with $\alpha_n\to\alpha$,
$$\Big({\upQ^{\{\bu\},\alpha_n}\wedge\refC}\Big) \xrightarrow[n\to\infty]{} \Big({\upQ^{\{\bu\},\alpha}\wedge\refC}\Big)$$ 
uniformly and, since $\cD$ is continuous with respect to the pointwise convergence of quasi-copulas, it follows that $\alpha\mapsto\cD\Big({\upQ^{\{\bu\},\alpha}\wedge\refC},\refC\Big)$ is continuous. 
In addition, we have that 
\begin{align}
\label{prescribedDistanceEq1}
\cD\Big({\upQ^{\{\bu\},M_d}\wedge\refC},\refC\Big)
	=\cD\Big({M_d\wedge\refC},\refC\Big) 
	= \cD\Big(\refC,\refC\Big)
	=0,
\end{align}
due to the fact that $C^*\lo M_d$. 
We now distinguish between two cases: 

\quad ($i$) Let $\delta\leq\cD\Big({\upQ^{\{\bu\},W_d}\wedge\refC},\refC\Big)$. 
Then, due to the monotonicity and continuity of the map $[W_d(\bu),M_d(\bu)]\ni\alpha\mapsto\cD\Big({\upQ^{\{\bu\},\alpha}\wedge\refC},\refC\Big)$ and \eqref{prescribedDistanceEq1} it holds that the set 
$$\mathcal{O}
	:=\Big\{\alpha\colon \cD\Big(\upQ^{\{\bu\},\alpha}\wedge\refC,\refC\Big) = \delta\Big\}$$
is non-empty and compact. 
Define $\alpha^* := \min\{\alpha\colon\alpha\in\mathcal{O}\}$. 
We will show that $\min\big\{Q(\bu)\colon Q\in\mathcal{Q}^{\cD,\delta}\big\}=\alpha^*$.
On the one hand, it holds that $\min\big\{Q(\bu)\colon Q\in\mathcal{Q}^{\cD,\delta}\big\}\leq\alpha^*$. 
Indeed, consider ${\upQ^{\{\bu\},\alpha^*}\wedge\refC}$ which is a quasi-copula and belongs to $\mathcal{Q}^{\cD,\delta}$ since $\alpha^*\in\mathcal{O}$. Then, we have that 
$$\Big({\upQ^{\{\bu\},\alpha^*}\wedge\refC}\Big)(\bu) = \min\{\alpha^*,C^*(u)\}  = \alpha^*,$$
using again that $\alpha^*\in\mathcal{O}$ and \eqref{prescribedDistanceEq1}.
Hence the inequality holds.
On the other hand, we will show now that the inequality cannot be strict by contradiction.
Assume there exists a quasi-copula $Q'\in\mathcal{Q}^{\cD,\delta}$ with $Q'(\bu)<\alpha^*$. 
Then it follows that 
\begin{align}
\begin{split}
\label{prescribedDistanceEq2}
\cD(Q',C^*) 
	& \geq \cD\big({Q'\wedge\refC},C^*\big) 
	  \geq \cD\Big({\upQ^{\{\bu\},Q'}\wedge\refC},C^*\Big)\\
	& \geq  \cD\Big({\upQ^{\{\bu\},\alpha^*}\wedge\refC},C^*\Big)
	= \delta,
\end{split}
\end{align}
where the first inequality follows from the min-stability of $\cD$, and the second and third ones from its monotonicity properties. 
However, since $Q'(\bu)\notin\mathcal{O}$ it follows that $ \cD\Big({\upQ^{\{\bu\},Q'}\wedge\refC},C^*\Big)\neq\delta$, hence \eqref{prescribedDistanceEq2} yields that $\cD\Big({\upQ^{\{\bu\},Q'}\wedge\refC},C^*\Big)>\delta$. 
This contradicts the assumption that $Q'\in\mathcal{Q}^{\cD,\delta}$, showing that indeed $\min\big\{Q(\bu)\colon Q\in\mathcal{Q}^{\cD,\delta}\big\}=\alpha^*$. Hence, the lower bound holds for $\delta\leq\cD\Big({\upQ^{\{\bu\},W_d}\wedge\refC},\refC\Big)$.

\quad ($ii$) Now, let $\delta>\cD\Big({\upQ^{\{\bu\},W_d}\wedge\refC},\refC\Big)$, then it follows that
$$\min\Big\{\alpha\in[W_d(\bu),M_d(\bu)]\colon\cD\Big({\upQ^{\{\bu\},\alpha}\wedge\refC},\refC\Big) \leq \delta\Big\} = W_d(\bu).$$
Moreover, since $\Big(\upQ^{\{\bu\},W_d}\wedge\refC\Big)\in\mathcal{Q}^{\cD,\delta}$ and every element in $\mathcal{Q}^{\cD,\delta}$ is bounded from below by $W_d$, it follows that $\min\big\{Q(\bu)\colon Q\in\mathcal{Q}^{\cD,\delta}\big\} = W_d(\bu)$. 
Hence, the lower bound holds in this case as well.

Finally, it follows again from \citep[Theorem 2.1]{Rodriguez_Ubeda_2004} that the bounds are quasi-copulas, which completes the proof.
\end{proof}

\begin{remark}
Let $C^*$ and $\cD$ be as in Theorem \ref{prescribedDistance}, and consider $\delta\in\mathbb{R}_+$. 
Then, the bounds $\underline{Q}^{\cD,\delta}$ and $\overline{Q}^{\cD,\delta}$ also apply to the set of copulas $\mathcal{C}^{\cD,\delta} := \{C\in\mathcal{C}^d\colon \cD(C,C^*) \leq \delta\}$, assuming that $\mathcal{C}^{\cD,\delta}\neq\emptyset$, that is
\begin{align}
\label{boundsDeltaCopulas}
\underline{Q}^{\cD,\delta} \lo C \lo \overline{Q}^{\cD,\delta},
\end{align}
for all $C\in\mathcal{C}^{\cD,\delta}$, due to the fact that $\mathcal{C}^{\cD,\delta}\subseteq\mathcal{Q}^{\cD,\delta}$.
\end{remark}

\begin{remark}
If $\cD$ is not symmetric, the set $\{Q\in\mathcal{Q}^d\colon \cD(Q,C^*) \leq \delta\}$ might not coincide with the set $\{Q\in\mathcal{Q}^d\colon \cD(C^*,Q) \leq \delta\}$. 
In this case the bounds on $\{Q\in\mathcal{Q}^d\colon \cD(C^*,Q) \leq \delta\}$ are provided by 
\begin{align*}
&\underline{Q}^{\cD,\delta}(\bu)=\min\Big\{\alpha\in[W_d(\bu),M_d(\bu)]\colon \cD\Big(\refC,\upQ^{\{\bu\},\alpha}\wedge\refC\Big) \leq \delta\Big\},\\
&\overline{Q}^{\cD,\delta}(\bu)=\max\Big\{\alpha\in[W_d(\bu),M_d(\bu)]\colon \cD\Big(\refC,\downQ^{\{\bu\},\alpha}\vee\refC\Big)\leq \delta\Big\}.
\end{align*}
\end{remark}

Many well-known statistical distances satisfy the requirements of Theorem \ref{prescribedDistance}. 
Typical examples are the Kolmogorov--Smirnov and the Cram\'er--von Mises distances, where
\begin{align*}
\cD_{\text{KS}}(Q,Q') := \sup_{\bu\in\mathbb{I}^d}|Q(\bu)-Q'(\bu)|	
	\quad \text{ and } \quad 
\cD_{\text{CM}}(Q,Q') := \int\nolimits_{\mathbb{I}^d} |Q(\bu)-Q'(\bu)|^2 \ud\bu.	
\end{align*}
The same holds for all $L^p$ distances with $p\geq 1$, where 
$$\cD_{L^p}(Q,Q') := \Big(\int\nolimits_{\mathbb{I}^d} |Q(\bu)-Q'(\bu)|^p \ud\bu\Big)^{\frac{1}{p}}.$$ 
Distances with these properties are of particular interest in the theory of minimum distance and minimum contrast estimation, where---as opposed to maximum likelihood methods---parameters of distributions are estimated based on a statistical distance between the empirical and the estimated distribution. 
These estimators have favorable properties in terms of efficiency and robustness; \textit{cf.} \citet[Chapter 2.8]{spokoiny2015}.

The computation of the bounds $\underline{Q}^{\cD,\delta}$ and $\overline{Q}^{\cD,\delta}$ in Theorem \ref{prescribedDistance} involves the solution of optimization problems, which can be computationally intricate depending on the distance $\cD$. 
An explicit representation of the bounds is thus highly valuable for applications. 
The following result shows that in the particular case of the Kolmogorov--Smirnov distance the bounds can be computed explicitly. 

\begin{lemma}
\label{explicitBounds}
Let $C^*$ be a $d$-copula, $\delta \in \R_+$, and consider the Kolmogorov--Smirnov distance $\cD_{\emph{KS}}$.
Then
\begin{align*}
\underline{Q}^{\cD_{\emph{KS}},\delta}(\bu) 
	= \max\big\{C^*(\bu)-\delta,W_d(\bu)\big\} 
\quad \text{ and } \quad
\overline{Q}^{\cD_{\emph{KS}},\delta}(\bu) 
	= \min\big\{C^*(\bu)+\delta,M_d(\bu)\big\}.
\end{align*}
\end{lemma}

\begin{proof}
Let us start with the lower bound $\underline{Q}^{\cD_{\text{KS}},\delta}$. 
Due to ${\upQ^{\{\bu\},\alpha}\wedge\refC} \preceq \refC$ for all $\alpha\in[W_d(\bu),M_d(\bu)]$, it holds that
\begin{align*}
\cD_{\text{KS}}\Big({\upQ^{\{\bu\},\alpha}\wedge\refC},\refC\Big) 
	&= \sup_{\bx\in\mathbb{I}^d} \Big|\big({\upQ^{\{\bu\},\alpha}\wedge\refC}\big)(\bx)-\refC(\bx)\Big| 
	 = \sup_{\bx\in\mathbb{I}^d} \Big\{ \refC(\bx)-\upQ^{\{\bu\},\alpha}(\bx) \Big\}.
\end{align*}
Since $\sup_{\bx\in\mathbb{I}^d} \big\{\refC(\bx)-\upQ^{\{\bu\},\alpha}(\bx)\big\} = 0$ when $\alpha> C^*(\bu)$, we can assume w.l.o.g. that the minimum is attained for $\alpha\leq C^*(\bu)$.
Hence
\begin{align*}
\min\Big\{\alpha \in [W_d(\bu),M_d(\bu)] &\colon \cD_{\text{KS}}\Big({\upQ^{\{\bu\},\alpha}\wedge\refC},\refC\Big) \leq \delta\Big\} \\
	&= \min\Big\{\alpha \in \mathbb [W_d(\bu),C^*(\bu)] \colon 
		\sup_{\bx\in\mathbb{I}^d} \Big\{\refC(\bx)-\upQ^{\{\bu\},\alpha}(\bx)\Big\} \leq \delta\Big\}.
\end{align*}
Then, using the definition of $\upQ^{\{\bu\},\alpha}$ in \eqref{bounds}, we obtain
\begin{align*}
\sup_{\bx\in\mathbb{I}^d} \Big\{\refC(\bx) - \upQ^{\{\bu\},\alpha}(\bx)\Big\}
	&= \sup_{\bx\in\mathbb{I}^d} \Big\{\refC(\bx) - \min\Big\{M_d(\bx),\alpha+\sum_{i=1}^d (x_i-u_i)^+\Big\}\Big\}\\
	&= \sup_{\bx\in\mathbb{I}^d} \Big\{\refC(\bx) - \alpha-\sum_{i=1}^d (x_i-u_i)^+\Big\}\\
	&= \sup_{\bx\in\mathbb{I}^d} \Big\{\refC(\bx) -\sum_{i=1}^d (x_i-u_i)^+\Big\} - \alpha
	 = C^*(\bu)-\alpha,
\end{align*}
where the second equality holds due to the fact that $C^*(\bx)-M_d(\bx)\leq 0$ for all $\bx \in \mathbb{I}^d$. 
Hence, we conclude that
\begin{align*}
\downQ^{\cD_{\text{KS}},\delta}(\bu) 
	&= \min\big\{\alpha \in \mathbb [W_d(\bu),C^*(\bu)] \colon C^*(\bu)-\alpha \leq \delta\big\} \\
	&= \min\big\{\alpha \in \mathbb [W_d(\bu),C^*(\bu)] \colon C^*(\bu)-\delta \leq \alpha\big\} 
	 = \max\big\{C^*(\bu)-\delta,W_d(\bu)\big\}.
\end{align*}
The proof for the upper bound $\overline{Q}^{\cD_{\text{KS}},\delta}$ is analogous, therefore omitted.
\end{proof}

Analogously to Theorem \ref{prescribedDistance}, one can also consider the situation where information on the survival copula is available. Note that each statistical distance that measures the discrepancy between quasi-copulas can easily be translated into a distance on quasi-survival functions, \textit{i.e.} if $\cD$ is a statistical distance on $\mathcal{Q}^d\times\mathcal{Q}^d$, then $(\widehat{Q},\widehat{Q}')\mapsto \cD\big(\widehat{Q}(\mathbf 1-\cdot),\widehat{Q}'(\mathbf 1-\cdot)\big)$ defines a distance on the set of survival copulas or quasi-survival functions. 

\begin{corollary}
\label{prescribedDistanceSurvival}
Let $C^*$ be a $d$-copula and $\cD$ be a statistical distance which is continuous with respect to the pointwise convergence of quasi-copulas, monotonic with respect to the upper orthant order and min/max-stable. 
Consider the set $\widehat{\mathcal{Q}}^{\cD,\delta} = \big\{\widehat{Q}\in\widehat{\mathcal{Q}}^d\colon \cD(\widehat{Q},\widehat{C}^*) \leq \delta\big\}$ for $\delta\in\mathbb{R}_+$.
Then
\begin{align*}
 \underline{\widehat{Q}}^{\cD,\delta}(\bu)
 	:=& \min\Big\{\alpha\in\mathbb S(\bu)\colon \cD\Big(\widehat{\overline{Q}}^{\{\bu\},\alpha}\wedge\refC,\refC\Big) \leq \delta\Big\} 
	 = \min\big\{C(\bu)\colon C\in\widehat{\mathcal{Q}}^{\cD,\delta}\big\} \\
 \widehat{\overline{Q}}^{\cD,\delta}(\bu)
 	:=& \max\Big\{\alpha\in\mathbb S(\bu)\colon \cD\Big(\widehat{\underline{Q}}^{\{\bu\},\alpha}\vee\refC,\refC\Big) \leq \delta\Big\} 
	 = \max\big\{C(\bu)\colon C\in\widehat{\mathcal{Q}}^{\cD,\delta}\big\}.
\end{align*}
\end{corollary}
The proof is analogous to the proof of Theorem \ref{prescribedDistance} and is therefore omitted.

\section{Numerical examples and illustrations}
\label{numerics}

In this section we apply the results deduced in the previous parts in order to derive bounds on the Value-at-Risk that account for additional  information on the dependence structure. 
In particular, we are able to include different types of partial dependence information that are both relevant for practical applications and have not been considered in the literature so far. 

The first example illustrates the improvement achieved by including extreme value information in the computation of the VaR bounds.
The setting is described in Section \ref{prescribedMax}, while a useful reduction argument is deferred to Appendix \ref{sec:reduction-arg}.

\begin{example}
\label{exExtremeValue}
We consider a homogeneous portfolio $\mathbf{X} = (X_1,\dots,X_6)$ where the marginals are Pareto-2 distributed, \textit{i.e.} $X_1,\dots,X_6\sim\text{Pareto}_2$, and analyze the improvement of the VaR bounds when additional information on the dependence structure is taken into account. 
In particular, we assume that the distributions $G_n$ of the maxima $\max_{j\in J_n} X_j$ are known for $J_1=\{1,2,3\}$ and $J_2=\{4,5,6\}$. 
In this case, it follows from Theorem \ref{boundMax} and equation \eqref{computeBoundsRA}, that
\begin{align}
\label{exExtremeValueEq1}
\begin{split}
& \sup_{(\alpha_1,\dots,\alpha_8)\in\underline{\mathcal{A}}} \underline{RA}(\alpha_1 Y_1,\alpha_2 Y_2,\alpha_3 X_1,\dots,\alpha_{8} X_6) \\
& \quad \leq \inf\Big\{\P(X_1+\cdots+X_6\leq s)\colon X_1,\dots,X_6\sim\text{Pareto}_2, Y_n \sim G_n, n = 1,2\Big\},
\end{split}
\end{align}
and analogously
\begin{align}
\label{exExtremeValueEq2}
\begin{split}
& \inf_{(\alpha_1,\dots,\alpha_8)\in\overline{\mathcal{A}}}  \overline{RA}(\alpha_1 Y_1,\alpha_2 Y_2,\alpha_3 X_1,\dots,\alpha_{8} X_6)\\
&\quad \geq \sup\Big\{\P(X_1+\cdots+X_6\leq s)\colon X_1,\dots,X_6\sim\text{Pareto}_2, Y_n \sim G_n, n = 1,2\Big\},
\end{split}
\end{align}
where $\underline{RA}(X,Y)$ and $\overline{RA}(X,Y)$ denote the lower and upper bound on $\P(X+Y\le\cdot)$ computed using the rearrangement algorithm (RA).

We have chosen to include the marginals $X_1,\dots,X_6$ in the optimization problems on the left hand side of \eqref{exExtremeValueEq1} and \eqref{exExtremeValueEq2} in order not to lose useful information about the marginal distributions, although the condition $\cup_n J_n=\{1,\dots,6\}$ of Theorem \ref{boundMax} is already satisfied. 
Note that the distribution of the maximum of every individual variable is trivially known and equals the respective marginal distribution; \textit{i.e.} $\max\{X_i\} = X_i\sim F_i$ for $i=1,\dots,d$. 

The solution of the optimization problems in \eqref{exExtremeValueEq1} and \eqref{exExtremeValueEq2} yields bounds on the VaR of the sum $X_1+\cdots+X_6$ when the distribution of the partial maxima is taken into account. 
Table \ref{tab:extremeValue} shows the confidence level $\alpha$ in the first column and the VaR bounds without additional information in the second column, \textit{i.e.} the unconstrained bounds. 
The third and fourth columns contain the improved VaR bounds that account for the extreme value information, as well as the improvement over the unconstrained bounds in percentage terms, \textit{i.e.} how much narrower the interval between the lower and the upper improved VaR bounds is relative to the same interval between the unconstrained bounds. 
In order to illustrate our method, we need to know the distribution of the partial maxima. 
To this end, we assume that the vectors $(X_1,X_2,X_3)$ and $(X_4,X_5,X_6)$ have the same Student-$t$ copula with equicorrelation matrices and two degrees of freedom, and numerically determine the distribution of $\max\{X_1,X_2,X_3\}$ and $\max\{X_4,X_5,X_6\}$. 
In the third column it is assumed that the pairwise correlations of $(X_1,X_2,X_3)$ and $(X_4,X_5,X_6)$ are equal to 0.9 and in the fourth column the pairwise correlations amount to 0.7 respectively.

The bounds in this table, both without and with additional information, have been computed using the RA, since the RA produces essentially sharp VaR bounds in the absence of additional information. 
Another possibility for the computation of the unconstrained bounds would be to use analytical results available \textit{e.g.} in \citet{Jakobsons_Han_Wang_2016,Puccetti_Rueschendorf_2013} and \citet{wang2013}.
We have refrained from doing so, since the computation of all bounds using the same algorithm makes the comparison of the results more credible, as numerical artifacts have been eliminated.

\begin{table}[h]
\begin{center}
\begin{tabular}{c|cc|ccc|ccc} 
\hline \hline
$\alpha$ & lower & upper & \shortstack{\vspace{0.1cm}\\ lower\\ improved} & \shortstack{\vspace{0.1cm}\\ upper\\ improved}  & \shortstack{\vspace{0.1cm}\\ impr.\\ \%} & \shortstack{\vspace{0.1cm}\\ lower\\ improved} & \shortstack{\vspace{0.1cm}\\ upper\\ improved}  & \shortstack{\vspace{0.1cm}\\ impr.\\ \%}\\ 
\hline
95\% & 3.8 & 47.8 & 3.8 & 39.5 & 19.7 & 4.9 & 44.8 & 9.1 \\ 
99\% & 4.9 & 114.0 & 11.0 & 96.1 & 22.0 & 12.4 & 107.8 & 12.5 \\ 
99.5\% & 5.2 & 163.7 & 16.1 & 138.5 & 22.7 & 18.0 & 155.1 & 13.5 \\ 
\hline \hline
\end{tabular}
\caption{Unconstrained and improved VaR bounds for the sum $X_1+\cdots+X_6$ with known distribution of partial maxima for different confidence levels.}
\label{tab:extremeValue}
\end{center}
\end{table}

The following observations ensue from this example: 
(i) The addition of partial dependence information allows to notably reduce the spread between the upper and lower bounds.
	Indeed, the bounds with additional information are finer than the unconstrained bounds resulting from the rearrangement algorithm, which are essentially sharp in this setting.  
	Nevertheless, the model risk is still not negligible. 
(ii) The level of improvement \textit{increases} with increasing confidence level $\alpha$.
	 This is in contrast to related results in the literature, see \textit{e.g.} \cite{bernard2015,bignozzi2015}, where the improvement typically decreases as the confidence level increases, and is an advantage of the present methodology.
(iii) The improvement is more pronounced in the high-correlation scenario, and for the lower bound. 
	  These two observations are in accordance with the related literature; \textit{e.g.} \citep{puccetti2016} report also a more pronounced improvement of the VaR bounds in the presence of strong positive dependence (especially in the tails), while \cite{bernard2015} report a more noticeable improvement of the lower relative to the upper VaR bound.
\end{example}

In the next example we combine the results of Section \ref{boundsOnCopula} with Proposition \ref{varBoundsMax} in order to derive improved bounds on the VaR of the maximum of risks over a class of copulas in the proximity of a reference copula.

\begin{example}
\label{exMaximumMinimum}
Let us consider a homogeneous portfolio of three risks $(X_1,X_2,X_3)$ where the marginals are again Pareto-2 distributed, \textit{i.e.} $X_1,X_2,$ $X_3\sim\text{Pareto}_2$. 
We assume that the reference copula $C^*$ is a Student-$t$ copula with equicorrelation matrix and two degrees of freedom, and are interested in computing bounds on the VaR over the class of models in the $\delta$-neighborhood of $C^*$ as measured by the Kolmogorov--Smirnov distance.
In other words, we consider the class 
$$\mathcal{C}^{\cD_{\text{KS}},\delta}:= \big\{C\in\mathcal{C}^d\colon \cD_{\text{KS}}(C,C^*) \leq \delta \big\},$$
and using Theorem \ref{prescribedDistance} and Lemma \ref{explicitBounds} we arrive at bounds on the copulas in $\mathcal{C}^{\cD_{\text{KS}},\delta}$. 

Then, we apply Proposition \ref{varBoundsMax} using the bounds $\underline{Q}^{\cD_{\text{KS}},\delta}$ and $\overline{Q}^{\cD_{\text{KS}},\delta}$ obtained above in order to compute bounds on the VaR of the maximum $\max\{X_1,X_2,X_3\}$ over the class of models in the vicinity of $C^*$. 
Table \ref{tab:distance1} shows the confidence level and the sharp unconstrained (\textit{i.e.} marginals-only) VaR bounds in the first two columns. 
The third, fourth and fifth column contain the upper and lower VaR bounds which use the information on the distance from $C^*$, for different levels of the threshold $\delta$, as well as the improvement over the unconstrained bounds in percentage terms, \textit{i.e.} how much narrower the interval between the lower and the upper improved VaR bounds is relative to the same interval between the unconstrained bounds. 
In this computation we assume that the pairwise correlation of the $t$-copula $C^*$ equals 0.9. 
The results are rounded to one decimal digit for the sake of legibility.

\begin{table}[h]
\begin{center}
\resizebox{\textwidth}{!}{%
\begin{tabular}{>{\centering}m{0.7cm}|c|cc|cc|cc} 
\hline \hline 
$\alpha$ & (lower : upper)& \shortstack{\vspace{0.1cm}\\ $\delta = 0.001$ \\ (lower : upper)}  & \shortstack{\vspace{0.1cm}\\ impr.\\ \%} &   \shortstack{\vspace{0.1cm}\\ $\delta = 0.005$\\ (lower : upper)}  & \shortstack{\vspace{0.1cm}\\ impr.\\ \%} & \shortstack{\vspace{0.1cm}\\ $\delta = 0.01$\\ (lower : upper)} & \shortstack{\vspace{0.1cm}\\ impr.\\ \%}\\ 
\hline
95\% & (1.4 : 6.8) & (3.6 : 4.6) & 81 & (2.5 : 4.7) & 59 & (2.3 : 5.0) & 50 \\ 
97\% & (2.0 : 9.1) & (4.8 : 6.2) & 78 & (3.5 : 6.7) & 55 & (3.2 : 7.7) & 37\\ 
99\% & (3.0 : 16.4) & (9.0 : 11.8) & 79 & (6.4 : 15.5) & 32 & (5.2 : 16.2) & 18 \\ 
\hline \hline
\end{tabular}
}
\caption{Unconstrained and improved VaR bounds for $\max\{X_1,X_2,X_3\}$ given a threshold on the distance from the reference $t$-copula $C^*$ with pairwise correlation equal to 0.9.}
\label{tab:distance1}
\end{center}
\end{table}

The next table is analogous to Table \ref{tab:distance1}, but this time weaker dependence is induced by the reference model, assuming that the pairwise correlations in the $t$-copula $C^*$ are equal to 0.6.

\begin{table}[h]
\begin{center}
\resizebox{\textwidth}{!}{%
\begin{tabular}{>{\centering}m{0.7cm}|c|cc|cc|cc} 
\hline \hline 
$\alpha$ & (lower : upper)& \shortstack{\vspace{0.1cm}\\ $\delta = 0.001$ \\ (lower : upper)}  & \shortstack{\vspace{0.1cm}\\ impr.\\ \%} &   \shortstack{\vspace{0.1cm}\\ $\delta = 0.005$\\ (lower : upper)}  & \shortstack{\vspace{0.1cm}\\ impr.\\ \%} & \shortstack{\vspace{0.1cm}\\ $\delta = 0.01$\\ (lower : upper)} & \shortstack{\vspace{0.1cm}\\ impr.\\ \%}\\ 
\hline
95\% & (1.4 : 6.8) & (3.5 : 5.3) & 67 & (1.5 : 5.6) & 24 & (1.4 : 5.8) & 19 \\ 
97\% & (2.0 : 9.1) & (4.8 : 7.2) & 66 & (2.3 : 7.8) & 23 & (2.0 : 8.8) & 4\\ 
99\% & (3.0 : 16.4) & (9 : 14) & 62 & (4.2 : 16.4) & 9 & (3.4 : 16.4) & 3 \\ 
\hline \hline
\end{tabular}
}
\caption{Unconstrained and improved VaR bounds for $\max\{X_1,X_2,X_3\}$ given a threshold on the distance from the reference $t$-copula $C^*$ with pairwise correlation equal to 0.6.}
\label{tab:distance2}
\end{center}
\end{table}

Let us point out that the bounds in Proposition \ref{varBoundsMax}, hence also in the second column of Tables \ref{tab:distance1} and \ref{tab:distance2}, are sharp when no dependence information is available, \textit{i.e.} when $\underline{Q} = W_3$ and $\overline{Q} = M_3$. 
This is due to the fact that $M_3$ is a copula and $W_3$ is pointwise best-possible. 

The observations made for the previous example are largely valid also in the present one, namely: 
(i) The addition of partial information reduces significantly the spread between the upper and lower bounds.
	This reduction is more pronounced as the threshold $\delta$ decreases; in other words, the more reliable the reference model, the more pronounced the reduction of model risk.
  	These results should be compared, qualitatively, with analogous results for the `trusted region' in \cite{bernard2015}.
(ii) The level of improvement decreases in this case, sometimes dramatically, with increasing confidence level $\alpha$.
	 In particular, for $\alpha = 99\%$ the improvement was small, especially for large values of $\delta$.
(iii) The improvement is more pronounced in the high-dependence scenario, with improvements over the sharp unconstrained bounds of up to 81\%.
\end{example}

\begin{remark}
The approach to compute VaR bounds over copulas in the vicinity of a reference model, as in Example \ref{exMaximumMinimum}, is applicable to statistical distances fulfilling the properties in Definition \ref{statisticalDistance}. 
Transportation distances, such as the Wasserstein distance, are typically not monotonic w.r.t. the orthant order, hence, our approach does not apply to them. 
A different method using neural networks to obtain risk bounds when information w.r.t. transportation distances is available, was recently presented by \citet*{eckstein2018}.
\end{remark}

\appendix
\section{On the computation of the bounds with known distribution of minima or maxima}
\label{sec:reduction-arg}

Let us first recall the setting of Section \ref{prescribedMax}, where we showed that
\begin{align*}
\underline{m}_{\mathcal{E},\max}(s) \le \P(X_1+\dots+X_d \le s) \le \overline{M}_{\mathcal{E},\max}(s);
\end{align*}
see Theorem \ref{boundMax}.
In order to compute the bounds $\underline{m}_{\mathcal{E},\max}(s)$ and $\overline{M}_{\mathcal{E},\max}(s)$, we first need to choose a method to estimate the probability $\P(\alpha_1Y_1+\cdots+\alpha_mY_m\leq s)$ for fixed $(\alpha_1,\dots,\alpha_m)$ in $\overline{\mathcal{A}}$ or $\underline{\mathcal{A}}$ and $Y_i\sim G_i$, $i=1,\dots,m$. 
This corresponds to a standard Fr\'echet problem over a class of distributions with fixed marginals. 
Thus, two approaches lend themselves naturally for this task: an approximation by the standard bounds given in \eqref{standardBounds} or by the rearrangement algorithm. 
Indeed, we can use the standard bounds in \eqref{standardBounds} to estimate
\begin{align*}
\max\Big\{0,\sup_{\mathcal{U}(s)}\sum_{i=1}^m G^-_i\Big(\frac{u_i}{\alpha_i}\Big)-m+1\Big\}
  &\le \P(\alpha_1Y_1+\cdots+\alpha_mY_m\leq s) \\
  &\qquad \le \min\Big\{1,\inf_{\mathcal{U}(s)}\sum_{i=1}^m G^-_i\Big(\frac{u_i}{\alpha_i}\Big)\Big\},
\end{align*}
where $\mathcal{U}(s) = \{(u_1,\dots,u_m)\in\mathbb{R}^m\colon u_1+\cdots+u_m=s\}$ and $G_i^-$ denotes the left-continuous version of $G_i$. 
Then, the bounds $\underline{m}_{\mathcal{E},\max}$ and $\overline{M}_{\mathcal{E},\max}$ are estimated by
\begin{align}
\label{computeBounds}
\begin{split}
&\underline{m}_{\mathcal{E},\max}(s) 
  \geq \sup_{(\alpha_1,\dots,\alpha_m)\in\underline{\mathcal{A}}} \max\Big\{0,\sup_{\mathcal{U}(s)}\sum_{i=1}^m G^-_i\Big(\frac{u_i}{\alpha_i}\Big)-m+1\Big\}, \\
&\overline{M}_{\mathcal{E},\max}(s) 
  \leq \inf_{(\alpha_1,\dots,\alpha_m)\in\overline{\mathcal{A}}} \min\Big\{1,\inf_{\mathcal{U}(s)}\sum_{i=1}^m G^-_i\Big(\frac{u_i}{\alpha_i}\Big)\Big\}.
\end{split}
\end{align}
Similarly, for fixed $(\alpha_1,\dots,\alpha_m)\in\underline{\mathcal{A}}$, the rearrangement algorithm allows us to approximate the bound
\begin{align}
\label{lowerRA}
\inf\big\{\P(\alpha_1Y_1+\cdots+\alpha_mY_m\leq s)\colon Y_n\sim G_n, n \in\mathcal{J}\big\},
\end{align}
while for $(\alpha_1,\dots,\alpha_m)\in\overline{\mathcal{A}}$ we can approximate
\begin{align}
\label{upperRA}
\sup\big\{\P(\alpha_1Y_1+\cdots+\alpha_mY_m\leq s)\colon Y_n\sim G_n, n \in\mathcal{J}\big\}.
\end{align}
To this end, we need to suitably discretize the variables $\alpha_1Y_1,\cdots,\alpha_mY_m$ and apply the rearrangement algorithm to the resulting matrix; for further details see \cite*{embrechts2013}. 
Denoting the lower bound in \eqref{lowerRA} computed by means of the rearrangement algorithm by $\underline{RA}(\alpha_1 Y_1,\dots,\alpha_m Y_m)$ and analogously the upper bound in \eqref{upperRA} by $\overline{RA}(\alpha_1 Y_1,\dots,\alpha_m Y_m)$, we thus obtain the following estimates:
\begin{align}
\label{computeBoundsRA}
\begin{split}
&\underline{m}_{\mathcal{E},\max}(s) 
  \geq \sup_{(\alpha_1,\dots,\alpha_m)\in\underline{\mathcal{A}}} \underline{RA}(\alpha_1 Y_1,\dots,\alpha_m Y_m), \\
&\overline{M}_{\mathcal{E},\max}(s) 
  \leq \inf_{(\alpha_1,\dots,\alpha_m)\in\overline{\mathcal{A}}} \overline{RA}(\alpha_1 Y_1,\dots,\alpha_m Y_m).
\end{split}
\end{align}
Let us stress that the RA has favorable numerical properties compared to the improved standard bounds. 
In particular, the bounds $\underline{RA}(\alpha_1 Y_1,\dots,\alpha_m Y_m)$ and $\overline{RA}(\alpha_1 Y_1,\dots,\alpha_m Y_m)$ can be computed very quickly for a reasonably fine discretization, thus the subsequent optimization over the set $\underline{\mathcal{A}}$ and $\overline{\mathcal{A}}$ can be performed much faster.

\section{Proof of Proposition \ref{varBoundsMax}}
\label{app:easy-proof}

\begin{proof}
Let $\varphi(x_1,\dots,x_d) = \max\{x_1,\dots,x_d\}$, then for any copula $C$ we have that 
\begin{align*}
\P_C(\max\{X_1,\dots,X_d\}<s) 
	= \P_C(X_1<s,\dots,X_d<s)
	= C(F_1(s),\dots,F_d(s)),
\end{align*}
using Sklar's Theorem for the last equality.
Hence, it follows immediately that
\begin{align*}
m_{\underline{Q},\max}(s) 
	&= \inf\big\{C(F_1(s),\dots,F_d(s))\colon \underline{Q}\preceq C\big\}
	\geq \underline{Q}(F_1(s),\dots,F_d(s)) \\
M_{\overline{Q},\max}(s)
	&= \sup\big\{C(F_1(s),\dots,F_d(s))\colon C\preceq \overline{Q}\big\}
	\leq \overline{Q}(F_1(s),\dots,F_d(s)).
\end{align*}
Moreover, since
\begin{align*}
\mathcal V^<_{\max} (s)
	&= \{ (x_1,\dots,x_d)\in\R^d: \max\{x_1,\dots,x_d\}<s\}\\
	&= \{ (x_1,\dots,x_d)\in\R^d: x_1<s,\dots,x_d<s\},
\end{align*}
we get from the improved standard bounds \eqref{eq:ISB} that 
\begin{align*}
\underline{m}_{\underline{Q},\max}(s) 
	&= \sup_{\mathcal V^<_{\max} (s)}\ \underline{Q}\big(F_1(x_1),\dots,F_{d}(x_{d})\big)
	= \underline{Q}(F_1(s),\dots,F_d(s)),
\end{align*}
where the last equality follows from the fact that $\underline{Q}$ is a quasi-copula, hence it is increasing in each component such that the supremum is attained at $(F_1(s),\dots,F_d(s))$.
The proof for the min operation is completely analogous and therefore omitted.
\end{proof}

\bibliographystyle{abbrvnat}
\bibliography{bib}

\end{document}